\documentclass[centertags,12pt]{amsart}
\usepackage{amssymb}
\usepackage{graphics}
\usepackage{graphicx}
\usepackage{latexsym}
\usepackage{color}

\makeatletter
\renewcommand{\subsection}{\@startsection
{subsection}{2}{0mm}{\baselineskip}{-0.25cm}
{\normalfont\normalsize\em}}
\makeatother

\newtheorem{theorem}{Theorem}[section]
\newtheorem{proposition}[theorem]{Proposition}
\newtheorem{corollary}[theorem]{Corollary}
\newtheorem{lemma}[theorem]{Lemma}

{\theoremstyle{definition}
\newtheorem{definition}[theorem]{Definition}
\newtheorem{example}[theorem]{Example}

}

\theoremstyle{remark}

\begin{document}

\title{On the order bounds for one-point AG codes}

\author{Olav Geil}
\author{Carlos Munuera}
\author{Diego Ruano}
\author{Fernando Torres}

  \begin{abstract}
The order bound for the minimum distance of algebraic geometry codes was originally defined for the duals of one-point codes and later generalized for arbitrary algebraic geometry codes. Another bound  of order type for the minimum distance of general linear codes, and  for codes from order domains in particular, was given in \cite{AG}. Here we investigate in detail the application of that bound to one-point algebraic geometry codes, obtaining a bound $d^*$ for the minimum distance of these codes. We establish a connection  between $d^*$ and the order bound and its generalizations. We also study the improved code constructions based on $d^*$. Finally we extend  $d^*$ to all generalized Hamming weights. 
  \end{abstract}

\thanks{Olav Geil and Diego Ruano are with the Department of Mathematical Sciences, Aalborg University, Fr. Bajersvej 7G, DK-9220 Aalborg, Denmark. 
Carlos Munuera is with the Department of Applied Mathematics, University of Valladolid, Avda Salamanca SN, 47014 Valladolid, Castilla, Spain.
Fernando Torres is with the Institute of Mathematics, Statistics and Computer Science, P.O. Box 6065, University of
Campinas, 13083-970, Campinas, SP Brazil}

\thanks{This work was supported in part 
by Danish National Science Research Council Grant FNV-21040368, Junta de CyL under grant VA065A07 and by Spanish Ministry for Science and Technology under grants MTM2007-66842-C02-01 and  MTM 2007-64704.}

\thanks{2000 {\em Mathematics Subject Classification:} Primary 94B27; Secondary 14G50, 14H55.}

\thanks{{\em Key words and phrases:} Linear codes, one-point algebraic geometry codes, minimum distance, Weierstrass semigroup, order bound.}

\maketitle


\section{Introduction}

Algebraic geometry codes, or AG codes, over the finite field ${\mathbb F}_q$ with $q$ elements are constructed from a (projective, non-singular, geometrically irreducible) algebraic curve ${\mathcal X} | {\mathbb F}_q$ and two rational divisors with disjoint support, $D=P_1+\dots+P_n$ and $G$ . The code $C(D,G)$ is defined as the image of the Riemann-Roch space ${\mathcal L}(G)$ by the evaluation at $D$  map $ev_D:{\mathcal L}(G)\rightarrow{\mathbb F}_q^n$, $ev_D(f)=(f(P_1),\dots,f(P_n))$, see Section 3 or \cite{Du,HLP,Sti}. The divisor $G$ is often taken as a multiple of a single point, $G=mQ$, with $Q\not\in\mbox{supp}(D)$. In this case $C(D,G)=C(D,mQ)$ is called {\em one-point} code.

Given a code $C(D,mQ)$ the first task is to compute its parameters: length, dimension and minimum distance. The length is obviously $n=\deg(D)$. In order to compute the dimension an important role is played by the Weierstrass semigroup at $Q$,
$$
H=H(Q):=\{ -v_Q(f) : f\in {\mathcal L}(\infty Q)\setminus \{ 0\}  \}=\{ h_1=0<h_2<\dots \}
$$
where $v_Q$ is the valuation at $Q$ and ${\mathcal L}(\infty Q)=\cup_{r=0,1,\dots} {\mathcal L}(rQ)$. In fact, if $h_i<n$ then the dimension of $C(D,h_iQ)$ is $i$. For $m\ge n$ this is no longer true in general, as the evaluation map $ev_D:{\mathcal L}(mQ)\rightarrow{\mathbb F}_q^n$ might have a non-trivial kernel, ${\mathcal L}(mQ-D)$. Thus we consider the set
$$
H^*=H^*(D,Q):=\{ m\in{\mathbb N}_0 : C(D,mQ)\neq C(D,(m-1)Q) \}.
$$
Knowing $H^*$ is equivalent to knowing the dimension of all codes $C(D,mQ)$. It is clear that $H^*$ consists of $n$ elements, that $H^*\subset H$ and that for $m<n$, $m\in H^*$ if and only if $m\in H$. 

Regarding the minimum distance $d=d(C(D,mQ))$ the simplest estimate is given by the Goppa bound, $d\ge n-m$. The Goppa bound does not give the true minimum distance in many cases. For example, it does not give any information when $m\ge n$. This problem can be solved by using the improved Goppa bound, $d\ge n-m+\gamma_{a+1}$,  where $a=\ell(mQ-D)$ is the {\em abundance} of $C(D,mQ)$. The drawback of this improved bound is that it is based on the gonality sequence $(\gamma_i)$ of the curve ${\mathcal X}$, see \cite{M}, which is difficult to compute. 

Besides   uniform bounds, some of the most interesting known bounds for $d$ are of order type. These bounds are based on obtaining different estimates for different subsets of codewords. They are successful if for each subset we can find estimates better than a uniform bound for all codewords, see \cite{Du2}. The original order bound $d_{ORD}$ (also called {\em Feng-Rao} bound) was introduced by Feng and Rao in \cite{FR} and by H\o holdt, van Lint and Pellikaan in \cite{HLP}.  It usually gives very good results, but it has the disadvantage that it can only be applied to the duals of one-point codes, which are not one-point codes in general. 
 A nice generalization of this bound for arbitrary AG codes was given by Beelen \cite{Beelen} and later improved by Duursma, Kirov and Park in a sequence of articles \cite{Du2,Du3,Du4}. 

Another  bound of order type for {\em general} linear codes was given in \cite{AG}.  This bound was applied to order domain codes and to one-point codes in particular. In the present work, we investigate in detail the case of one-point codes, obtaining a bound $d^*$. This bound was already present in \cite{AG} (Proposition 37)   but here we state it explicitly, by showing how to compute $d^*$  from the set $H^*$ defined above. Besides we investigate the connection to the order bound. 
 We show that $d^*$ is a special case of the Beelen and Duursma-Kirov-Park generalized bouds.  Since it can happen that the generalized order bounds give different results than the original one,  we also investigate the connection of $d^*$ to the original order bound $d_{ORD}$. We show that when both can be applied -namely when the dual of a one-point code is isometric to a one-point code- then both coincide.  Furthermore we investigate how to construct improved codes from $d^*$ and how to extend $d^*$ to all generalized Hamming weights. These problems have never been treated in the aforementioned works of Beelen and Duursma-Kirov-Park. 
  Thus the main purpose of this article is not to present a new or better bound, but (i) to make the conection between the Andersen-Geil bound and the order bounds for AG one-point codes,  (ii) to emphasize the possibility of manage the order bound entirely in the language of one-point evaluation codes and Weierstrass semigroups; (iii) to study how to construct improved codes; and (iv) to extend $d^*$ to all generalized Hamming weights.  

The paper is structured in 5 sections: In Section \ref{se:bo} we briefly recall the bound for the minimum distance of linear codes from \cite{AG}  as well as the main facts and definitions  we need.  We introduce the bound $d^*$ for one-point codes in Section \ref{se:boop}, 
 where we  also show   the connection with the generalized order bounds of Beelen and Duursma-Kirov-Park.   
We also deal with improved codes, whose construction  becomes now very easy. 
Some worked examples where we show how to compute $d^*$ are included.
In Section \ref{se:reb} we compare the bound $d^*$ to the strict order bound (that is the original order bound $d_{ORD}$ with respect to the evaluation map $ev_D$), showing that when both can be applied then they give the same result. Furthermore, we continue our study of improved codes. Finally in Section \ref{se:ghw} we extend $d^*$ to all generalized Hamming weights.

\section{The bound from \cite{AG} for the minimum distance of linear codes} \label{se:bo}

For the convenience of the reader, we begin with a brief explanation of some results from \cite{AG}. Let ${\mathcal B} =\{ {\mathbf b}_1, \dots, {\mathbf b}_n\}$ be a basis of ${\mathbb F}_q^n$. We consider the codes $C_0=(0)$, and for $i=1,\dots,n$, 
$$
C_i=\langle {\mathbf b}_1,\dots,{\mathbf b}_i \rangle .
$$ 
Associated to these codes we consider the (valuation-like) map $\nu:{\mathbb F}_q^n \rightarrow \{0,\dots,n\}$
defined by $\nu({\mathbf v})=\min \{ i : {\mathbf v} \in C_i\}$. 

  \begin{lemma}\label{properties} 
Let ${\mathbf v}_1,\dots,{\mathbf v}_m \in {\mathbb F}_q^n$. Then 
\begin{enumerate}
\item[\rm(a)] $\nu({\mathbf v}_1+\dots+{\mathbf v}_m) \leq \max \{\nu({\mathbf v}_1), \dots, \nu({\mathbf v}_m)\}$. If there exists $j$ such that $\nu({\mathbf v}_i) <\nu({\mathbf v}_j)$ for all $i \neq j$, then equality holds. 
\item[\rm(b)] $\dim (\langle {\mathbf v}_1,\dots,{\mathbf v}_m \rangle ) \geq \# \{\nu({\mathbf v}_1),\dots,\nu({\mathbf v}_m)\}$. 
Conversely, if $D\subseteq {\mathbb F}_q^n$ is a linear subspace of dimension $m$, then  there exists a basis $\{{\mathbf v}_1,\dots,{\mathbf v}_m\}$ of $D$ such that $\# \{\nu({\mathbf v}_1),\dots,\nu({\mathbf v}_m)\}=m$.
\end{enumerate} 
  \end{lemma} 
  \begin{proof}
(a) is clear. (b) Assume $\#\{\nu({\mathbf v}_1),\dots,\nu({\mathbf v}_m)\} = t$ and $\nu({\mathbf v}_1)<\dots <\nu({\mathbf v}_t)$. If $\lambda_1{\mathbf v}_1+\dots+\lambda_t{\mathbf v}_t=0$ then $0=\nu({\mathbf 0})=\nu(\lambda_1{\mathbf v}_1+\dots+\lambda_t{\mathbf v}_t)=\max \{ \nu({\mathbf v}_i) : \lambda_i \neq 0\}$. 
By (a) this implies $\lambda_1= \dots =\lambda_t=0$. Conversely write $D_i=D\cap C_i$. For all $i=1,\dots,n$, it holds that  $D_i=D_{i-1} \oplus (D\cap \langle {\mathbf b}_i \rangle )$, hence $\dim (D_{i-1})\le \dim (D_i)\le \dim (D_{i-1})+1$ and the last inequality is an equality precisely $m$ times. If $D_i\neq D_{i-1}$, take a vector ${\mathbf v}_i\in D_i\setminus D_{i-1}$. Then $\# \{\nu({\mathbf v}_1),\dots,\nu({\mathbf v}_m)\}=m$ and according to (b), $\{{\mathbf v}_1,\dots, {\mathbf v}_m\}$ is a basis of $D$. 
  \end{proof}

For ${\mathbf c} \in {\mathbb F}_q^n$, ${\mathbf c} \neq 0$, we consider the space $V({\mathbf c})= \{ {\mathbf v} \in {\mathbb F}_q^n : \mbox{supp} ({\mathbf v}) \subseteq\mbox{supp}({\mathbf c})\} = \{ {\mathbf v}*{\mathbf c} : {\mathbf v} \in {\mathbb F}_q^n\}$, where the component-wise product is defined as usual: ${\mathbf v}*{\mathbf c}=(v_1c_1,\dots,v_nc_n)$. Clearly $\dim (V({\mathbf c}))=\mathrm{wt}({\mathbf c})$, where $\mathrm{wt}({\mathbf  c})$ denotes the weight of ${\mathbf c}$. Now consider in $\{ 1,\dots,n\}^2$ the order $(r,s) < (i,j)$ if and only if $r \leq i$, $s \leq j$ and $(r,s) \neq (i,j)$. A pair $({\mathbf b}_i,{\mathbf b}_j)$ is called {\em well-behaving} if $\nu({\mathbf b}_r*{\mathbf b}_s)<\nu({\mathbf b}_i*{\mathbf b}_j)$ for all $(r,s)<(i,j)$. For $i=1,\dots, n$, define
$$
\Lambda_i= \{ {\mathbf b}_j \in {\mathcal B} : ({\mathbf b}_i,{\mathbf b}_j) \mbox{ is well-behaving}\}  .
$$ 

Since we can write ${\mathbf c}=\lambda_1{\mathbf b}_1+\dots+\lambda_{\nu({\mathbf c})}{\mathbf b}_{\nu({\mathbf c})}$ with $\lambda_{\nu({\mathbf c})} \neq 0$, then for ${\mathbf b}_j \in \Lambda_{\nu({\mathbf c})}$ we have 
$$
\nu({\mathbf c}*{\mathbf b}_j)=\nu(\sum_{i=1}^{\nu({\mathbf c})} \lambda_i{\mathbf b}_i*{\mathbf b}_j)=\nu({\mathbf b}_{\nu({\mathbf c})}*{\mathbf b}_j).
$$

  \begin{proposition}\label{agbound}
Let  ${\mathbf c} \in {\mathbb F}_q^n$. If ${\mathbf c} \neq {\mathbf 0}$ then $\mathrm{wt}({\mathbf c}) \geq \# \Lambda_{\nu({\mathbf c})}$.
  \end{proposition}
  \begin{proof}
We have $\mathrm{wt}({\mathbf c})=\dim (V({\mathbf c})) \geq \dim (\langle {\mathbf c}*{\mathbf b}_1,\dots,{\mathbf c}*{\mathbf b}_n \rangle) \geq \#\{\nu({\mathbf c}*{\mathbf b}_1),\dots, \nu({\mathbf c}*{\mathbf b}_n)\} \geq \#\{\nu({\mathbf c}*{\mathbf b}_j) : j \in \Lambda_{\nu({\mathbf c})}\}=\#\{\nu({\mathbf b}_{\nu({\mathbf c})}*{\mathbf b}_j) : j \in \Lambda_{\nu({\mathbf c})}\}= \#\Lambda _{\nu({\mathbf c})}$.
  \end{proof}

  \begin{theorem} 
For $i=1,\dots,n$, the true minimum distance of $C_i$, satisfies $d(C_i) \geq  \min \{ \#\Lambda_r : r \leq i\}$.  
  \end{theorem}

This bound can be applied to an arbitrary linear code $C$, just by including it into an increasing chain of codes  $C_1 \subset\dots\subset C_{k-1} \subset C \subset C_{k+1}\subset\dots \subset C_n={\mathbb F}_q^n$. Such a chain  is quite natural for one-point codes.

\section{A bound for the minimum distance of one-point codes} \label{se:boop}

\subsection{The bound}

Let ${\mathcal X}$ be a (projective, non-singular, geometrically irreducible algebraic) curve of genus $g$ defined over the finite field ${\mathbb F}_q$. We construct one-point codes from $\mathcal X$ in the usual way. Let $Q,P_1,\dots,P_n$ be different rational points in ${\mathcal X}$. Let $v=-v_Q$, where $v_Q$ is the valuation at $Q$, and consider the spaces ${\mathcal L} (mQ)$ and the algebra ${\mathcal L}(\infty Q)=\cup_{r=0,1,\dots} {\mathcal L} (rQ)$. Let $D=P_1+\dots+P_n$ and $ev=ev_D: {\mathcal L}(\infty Q) \rightarrow {\mathbb F}_q^n$ be the evaluation map at $D$. The one-point codes $C(D,mQ)$ arising from ${\mathcal X}, D$ and $Q$ are defined as the images of the sets ${\mathcal L} (mQ)$ by $ev$, that is $C(D,mQ)=ev({\mathcal L} (mQ))$. Note that $C(D,(n+2g-1)Q)={\mathbb F}_q^n$, hence we can restrict ourselves to $0\le m\le n+2g-1$.

Let $C=C(D,mQ)$. We shall apply to $C$ the bound from Section \ref{se:bo} with respect to the sequence of codes $C_1 \subset C_2 \subset \dots \subset C_n$, obtained from the sequence $( C(D, mQ) )_{m=0,\dots,n+2g-1}$ by deleting the repeated codes. Thus the map $\nu$ can be written as  
$$
\nu({\mathbf v})=\min \{\dim (C(D,mQ)) : {\mathbf v} \in C(D,mQ)\}.
$$
From now on, unless explicitly said, we restrict ourselves to codes with length $n > 2g+2$.

  \begin{lemma} \label{le:3.1}
For $f \in {\mathcal L}(\infty Q)$ we have $\nu(ev(f)) \leq \dim (C(D,v(f)Q))$. If $C(D,v(f)Q) \neq C(D,(v(f)-1)Q)$ then  equality holds, $\nu(ev(f)) = \dim (C(D,v(f)Q))$. 
  \end{lemma}
  \begin{proof}
The first statement is clear since $f\in {\mathcal L} (v(f)Q)$ and hence $ev(f) \in C(D,v(f)Q)$. For the second one, note that if $m=v(f)$, then ${\mathcal L} (mQ)= {\mathcal L} ((m-1)Q) + \langle f \rangle$, and hence $C(D,mQ)=C(D,(m-1)Q)+\langle ev(f) \rangle$. Thus $ev(f) \in C(D,mQ) \setminus C(D,(m-1)Q)$. 
  \end{proof} 

Note that it is not true in general that $\nu(ev(f)) = \dim (C(D,v(f)Q))$ because $ev$ only depends on the points $P_1,\dots,P_n$, and thus $ev(f)$ might be equal to $ev(g)$ with $g\in C(D,(v(f)-1)Q)$. For example, take a non-constant function $f \in {\mathcal L}(\infty Q)$. Then $v(f^q)=qv(f)$ but $ev(f^q)=ev(f)$.

Let $H=H(Q)= \{ h_1=0<h_2<\dots\}$ be the Weierstrass semigroup of $Q$. As we know, this is a numerical semigroup of finite genus $g$. Let $l_1,\dots,l_g$ be the gaps of $H$. Let us consider the set $H^*$ defined in the Introduction, namely 
$$
H^*=H^*(D,Q):=\{ m\in{\mathbb N}_0 : C(D,mQ)\neq C(D,(m-1)Q) \}.
$$
It is clear that $H^*$ consists of $n$ elements. Let us write $H^*= \{ m_1,\dots,m_n\}$. It is also clear that $H^*\subset H$ and for $m<n$ it holds that $m\in H^*$ if and only if $m\in H$. The following results may be useful for computing $H^*$. Remember that for a divisor $E$, $\ell(E)$ stands for the dimension of ${\mathcal L}(E)$.

  \begin{proposition}\label{ell}
$H^*= \{ m\in H : \ell(mQ-D)= \ell((m-1)Q-D)\}$.
  \end{proposition} 
  \begin{proof}
If $m<n$ then $C(D,mQ)\neq C(D,(m-1)Q)$ if and only if $m\in H$ that is if and only if $m\in H^*$. If $m\geq n$ then the kernel of the evaluation map restricted to $\mathcal L (mQ)$ is $\ker (ev|_{{\mathcal L}(mQ)})={\mathcal L}(mQ-D)$. Since $m-1,m\in H$, then $C(D,mQ)\neq C(D,(m-1)Q)$ if and only if both kernels are equal.
  \end{proof}

Thus, for $m\ge n$, and since $\ell((n+2g-1)Q-D)=g$ and $H$ has $g$ gaps, we conclude that $g$ elements of $\{ n,\dots, n+2g-1 \}$ belong to $H^*$ while the other $g$ elements do not.

  \begin{corollary}\label{m+h}
Let $m\ge n$. If $m\not\in H^*$ then for all $h\in H$ it holds that $m+h\not\in H^*$. 
  \end{corollary}
  \begin{proof}
If $m\not\in H^*$ then there exists a non-zero function $f\in {\mathcal L}(mQ-D)\setminus {\mathcal L} ((m-1)Q-D)$. Take a function $\phi\in {\mathcal L} (hQ)$ such that $v(\phi)=h$. Then  $f\phi\in {\mathcal L}((m+h)Q-D)\setminus {\mathcal L} ((m+h-1)Q-D)$, and hence $m+h\not\in H^*$.
  \end{proof}
 
  \begin{corollary}\label{sim}
If the divisors $D$ and $nQ$ are linearly equivalent, $D\sim nQ$, then  $H^*\cap\{ n,\dots,n+2g-1\}=\{ n+l_1,\dots,n+l_g \}$, hence $H^*= (H\cap \{ 1,\dots,n-1 \})\cup\{ n+l_1,\dots,n+l_g\}$.
  \end{corollary}
  \begin{proof}
If $D\sim nQ$ then  $n\not\in H^*$ and hence, according to Corollary \ref{m+h}, $n=n+h_1,\dots,n+h_g\not\in H^*$. The statement follows by cardinality reasons.
  \end{proof}

Let $f\in {\mathcal L}(\infty Q)$. If $v(f)\in H^*$ then, by Lemma \ref{le:3.1}, we have $\nu(ev(f))=\dim (C(v(f)))$. For $i=1,\dots,n$, let $f_i\in {\mathcal L}(\infty Q)$ be such that $v(f_i)=m_i$. Thus, according to Lemma \ref{properties} (b),   ${\mathcal B}=\{ ev(f_1),\dots,ev(f_n) \}$ is a basis of ${\mathbb F}_q^n$ and the sequence of codes $(C_i)$ is given by
$$
C_i=\langle ev(f_1),\dots,ev(f_i) \rangle=C(D,m_iQ), \; i=1,\dots,n.
$$
Our sequence $(C(D,m_iQ))$ does not contain the code $C_0=(0)$. If we want to include it (see Section 4 for example) we simply take $m_0=-1$ and $C(D,m_0Q)=(0)$.

  \begin{proposition}\label{wbp}
If $m_i+m_j\in H^*$ then $(ev(f_i),ev(f_j))$ is a well behaving pair.
  \end{proposition}
  \begin{proof}
For $\phi_1,\phi_2\in {\mathcal L}(\infty Q)$ we have that $v(\phi_1\phi_2)=v(\phi_1)+v(\phi_2)$. 
If $m_i+m_j\in H^*$ then $\nu(ev(f_i)*ev(f_j))=\nu(ev(f_if_j))=\dim (C(D,v(f_if_j)Q))=\dim (C(D,(m_i+m_j)Q))$. If $(r,s)<(i,j)$ then $v(f_rf_s)<v(f_if_j)$ and hence $\nu(ev(f_r)*ev(f_s))=\nu(ev(f_rf_s))< \dim (C(D,(m_i+m_j)Q))$.
  \end{proof}

Thus from the bound in Section \ref{se:bo} we get a bound for one-point codes as follows. For $i=1,\dots,n$, consider the sets
$$
\Lambda^*_i=\{ m\in H^* : m=m_i+m_j \mbox{ with }  m_j\in H^* \}.
$$
If $m\in m_i+H\setminus H^*$ then $m=m_i+h$ for some $h\in H\setminus H^*$ and thus  $m\not\in H^*$ according to Corollary \ref{m+h}. Thus the sets $\Lambda^*_i$ can also be written as $\Lambda^*_i=\{m\in H^* : m-m_i\in H \}=(m_i+H)\cap H^*$.
According to  Propositions \ref{agbound} and \ref{wbp}, we have  that $\mathrm{wt}({\mathbf c})\ge \# \Lambda^*_r$ for all ${\mathbf c}\in C(D,m_rQ)\setminus C(D,m_{r-1}Q)$. Define
$$
d^*(i):=\min \{ \# \Lambda^*_r : r\le i  \}.
$$ 
Then $d(C(D,m_iQ))\geq d^*(i)$,  or equivalently 

  \begin{theorem} \label{asterisk}
For a non-negative integer $m$, we have $d(C(D,mQ))\geq d^*(\dim (C(D,mQ)))$.
  \end{theorem}

We call this inequality the {\em $d^*$ bound} for one-point codes.
Let us remember that the classical bound on the minimum distance of an code is given by the Goppa estimate $d(C(D,mQ))\ge d_G(C(D,mQ)):=n-m$. $d^*$ improves the Goppa bound as the next result shows (see also Proposition 37 in \cite{AG}). The first element in $H\setminus H^*$ is denoted by $\pi=\pi(H)$. Note that $\pi\ge n$. 

  \begin{proposition}\label{goppa}
For all $i=1,\dots,n$, we have $d^*(i)\ge d_G(C(D,m_iQ))$. If $m_i<\pi-l_g$ then  equality holds, $d^*(i)= d_G(C(D,m_iQ))$.
  \end{proposition}
  \begin{proof}
For the first statement it suffices to show  that $\#(H^*\setminus \Lambda_r^*)\le m_r$ for all $r$. Since $\Lambda^*_i=(m_i+H)\cap H^*$, we have $H^*\setminus \Lambda_i^*\subseteq H\setminus (m_i+H)$ and this follows from the fact that $\#(H\setminus (m_r+H))=m_r$ (see \cite{HLP}, Lemma 5.15).  If $m_i+l_g<\pi$, then all elements in $H\setminus (m_i+H)$ are smaller than $\pi$ and hence $H^*\setminus \Lambda_i^*= H\setminus (m_i+H)$. 
  \end{proof}

   \subsection{$d^*$ and the generalized order bounds of Beelen and Duursma-Kirov-Park}\label{sec:BDKP}
   
The bound $d^*$ can also be obtained from the generalized order bounds of Beelen and Duursma-Kirov-Park. Let us show first how to get $d^*$ from the Beelen generalized order bound $d_B$ stated in \cite{Beelen}. Let $m_i\in H^*$ and consider the code $C(D,m_iQ)$.  The Beelen bound applies to the duals of evaluation codes. Thus, let $W$ be a canonical divisor with simple poles and residue 1 at all points  $P\in {\rm supp}(D)$ and let $G=D+W-m_iQ$. It is well known that $C(D,m_iQ)=C(D,G)^{\perp}$ (see \cite{Sti}). By using the notation as in \cite{Beelen}, for $r=0,1,2,\dots$, consider the divisors 
$$
F^{(r)}:=G+rQ = F^{(r)}_1+F^{(r)}_2=: (D+W)+((r-m_i)Q).
$$
Note that all the divisors $F^{(r)},F^{(r)}_1,F^{(r)}_2$ above have support disjoint from $D$. For a divisor $E$, let $H(Q,E)$ be the Weierstrass set of $Q$ relative to $E$,
$$
H(Q,E)=-v_Q \left( \bigcup_{\deg(E+sQ)\ge 0} \mathcal{L}(E+sQ) \setminus\{ 0\} \right).
$$ 
In  our case, for all $r=0,1,\dots$, we have $H(Q, F^{(r)}_2)=H(Q,(r-m_i)Q)=H(Q,0)=H$, the usual Weierstrass semigroup of $Q$. The Beelen bound states that
$$
d(C(D,m_iQ))\ge \min \{ \# N(F_1^{(r)},F_2^{(r)}) : r=0,1,\dots\}
$$
where 
\begin{eqnarray*}
N(F_1^{(r)},F_2^{(r)}) &=&\{ (t,s) : t\in H(Q,F_1^{(r)}), s\in H(Q,F_2^{(r)}), t+s=v_Q(G)+1  \} \\
                       &=&\{ (t,s) : t\in H(Q,D+W), s\in H, t+s=1-m_i  \}.
\end{eqnarray*}
According to the Rieman-Roch theorem, for an integer $m$ it holds that $1-m\in H(Q,D+W)$ if and only if $\ell (mQ-D)= \ell ((m-1)Q-D)$. Thus  for $m\in H$ the conditions $m\in H^*$ and $1-m\in H(Q,D+W)$ are equivalent. Consequently
\begin{eqnarray*}
\# N(F_1^{(r)},F_2^{(r)}) &=& \# \{ s\in H : 1-(s+m_i) \in H(Q,F_1^{(r)})  \} \\
                          &=& \# \{ s\in H : s+m_i \in H^* \}
\end{eqnarray*}
as $s\in H$ implies $s+m_i\in H$. 
Finally observe that while the sets $\Lambda^*_i$ and $\{ s\in H : s+m_i \in H^* \}$ count different objects, they are of the same cardinality: the map $m\mapsto m+m_i$ gives a bijection from $\Lambda^*_i$ to $\{ s\in H : s+m_i \in H^* \}$. Thus, for one-point codes, the bound $d^*$ can be seen as a particular case of the Beelen bound $d_B$, relative to the choice of $Q,Q,\dots $ as infinite sequence of points not in ${\rm supp}(D)$ and the divisors  
$F_1^{(r)}=D+W ,F_2^{(r)}=(r-m_i)Q$. In particular in may happen that $d^*<d_B$ (for an accurate choice of the infinite sequence of points and the divisors $F_1^{(r)},F_2^{(r)}$), 
in the same way as it may happen that $d_{ORD}<d_B$ (see Example 8 of \cite{Beelen}).

Let us show briefly how to obtain $d^*$ from the generalized order bound of Duursma, Kirov and Park. Consider again the code $C(D,m_iQ)$. In the formulation of \cite{Du2,Du3,Du4}, if $\mathbf{c}\in C(D,m_iQ)\setminus C(D,m_{i-1}Q)$, then 
$$
\mathrm{wt}({\mathbf c}) \geq \# (\Delta_Q(D-m_iQ)\cap \{ (m-m_i)Q : m\ge m_i \} )
$$ 
where for a divisor $E$, $\Delta_Q(E)$ is defined as
$$
\Delta_Q(E)=\{ A : \mathcal{L}(A)\neq \mathcal{L}(A-Q),  \mathcal{L}(A-E)\neq \mathcal{L}(A-Q-E) \}.
$$
The same argument as in the case of $d_B$ proves that the sets $\Lambda^*_i$ and $(\Delta_Q(D-m_iQ)\cap \{ (m-m_i)Q : m\ge m_i \} )$ are of the same cardinality. 
This shows that  $d^*$ can also be  obtained from the extended Duursma-Kirov-Park order bound.

 On the other hand, the choice of the sets $\Lambda^*_i$ (instead of the counting made in the Beelen and Duursma-Kirov-Park bounds) has some technical advantages. Firstly it does not involve more divisors that the ones naturally associated to the code $C(D, mQ) $. And secondly, in contrast to what happens with those bounds,  $d^*$ allows us to study improved codes very easily. Also it allows us to extend the same idea to all generalized Hamming weights (see Section \ref{se:ghw}). In fact, for these two problems  $d^*$ works even better than the original order bound $d_{ORD}$. As discussed in Section \ref{se:reb}, $d^*$ extends exactly $d_{ORD}$ to one-point codes.

   \subsection{Improved codes}\label{Sect_improved}

Let $\delta$ be an integer, $0<\delta\le n$.   
In the same way as the order bound allows us to construct codes with designed minimum distance  $\delta$ and dimension as large as possible, see \cite{HLP}, the bound $d^*$ shows how to construct similar codes from  sequences $(C(D,m_i Q))$, see \cite{AG}.  Specifically, given $\delta$ let us consider the {\em improved code}
$$
C(D,Q,\delta)=\langle \{ ev(f_i) : \#\Lambda^*_i\ge \delta  \}  \rangle
$$
where $f_i\in {\mathcal L}(\infty Q)$ with $v(f_i)=m_i$.  From Lemma \ref{properties} (a), and the discussion before Theorem  \ref{asterisk}, it is clear that the minimum distance of $C(D,Q,\delta)$ is at least $\delta$. 

The sequence $(\Lambda^*_i)$ is said to be {\em monotone} for $\delta$ if for every $i,j$ such that $\#\Lambda^*_i\ge \delta$ and $\#\Lambda^*_j<\delta$ we have that $i<j$. If $(\Lambda^*_i)$ is  monotone for $\delta$ it is clear that $C(D,Q,\delta)$ is a usual one-point code, so improved codes only improve one-point codes for those $\delta$ for which the sequence is not monotone. In this case the code $C(D,Q,\delta)$ depends on the choice of the set $\{ f_1,\dots,f_n \}$. In fact, if $\#\Lambda^*_i=\delta$ and $\#\Lambda^*_j<\delta$ for some $j<i$, then $v(f_i+f_j)=v(f_i)$ but in general $ev(f_j)\not\in C(D,Q,\delta)$, hence $ev(f_i+f_j)\not\in C(D,Q,\delta)$. Thus we have a collection of improved codes with designed distance $\delta$, depending on the collection of sets $\{ f_1,\dots,f_n \}$.

   \subsection{Worked examples}

We compute $H^*$ for some examples.

  \begin{example} \label{castle} (Codes on Castle curves) 
A curve $\mathcal X$  defined over  ${\mathbb F}_{q}$ is said to be {\em Castle} if there is a rational point $Q$ such that the Weierstrass semigroup at $Q$, $H=H(Q)$, is symmetric and $qh_2+1=\# {\mathcal X}({\mathbb F}_{q})$ (where $h_2$ is the first nonzero element of $H$). If $D$ is the sum of all rational points of $\mathcal X$ except $Q$, the codes $C(D,mQ)$ are called Castle codes, see \cite{MST}. It is simple to see that for Castle curves we have $D\sim nQ$, hence 
$H^*\cap\{ n,\dots,n+2g-1\}=\{ n+l_1,\dots,n+l_g \}$ according to Proposition \ref{sim}. In Section 4 we shall see that, being the semigroup $H$ symmetric, we have $H^*=H\setminus (n+H)$. Recall that the family of Castle codes includes  Hermitian, generalized Hermitian, Norm-trace, Suzuki, Ree and many of the most known codes.  To study a concrete example, let us consider the Suzuki curve $\mathcal X$ over ${\mathbb F}_{8}$ (see \cite{MST} again). This curve has genus $g=14$ and 65 rational points. A plane model of $\mathcal X$ is given by the equation  $Y^8Z^2-YZ^{9} = X^2(X^8-XZ^7)$.  This model is non-singular except at the point $(0 : 1 : 0)$. Being this singularity uni-branched, the unique point $Q$ lying over $(0 : 1 : 0)$ is rational. Let us consider the codes $C(D,mQ)$, where $D$ is the sum of all rational points of $\mathcal X$ except $Q$. The Weierstrass semigroup at $Q$ is known to be $H=\langle 8,10,12,13\rangle$. A straightforward computation gives the sequence $(\#\Lambda^*_i)$: $(64$, $56$, $54,50,49,48,46,44,43,42,41,40,38,36,35,34,33,32,31,30,29,28,28,26,26,24,23,22,21$,
\linebreak
$20,21,18,20,16,18,16,14,13,14,10,14,8,13,10,10,9,9,6,9,8,4,6,5,5,4,6,5,3,2,3,3,2$, $\;$ 
$1,1)$. 
This sequence is monotone for $\delta=3,5,6,9,13,14,18,20,21$. For example the code $C(D,70Q)$ has dimension 55 and distance at least 4 (that is $d^*(55)=4$), whereas $C(D,Q,4)$ has dimension 57.
  \end{example}

  \begin{example}\label{ej} (Two families of codes from a curve over ${\mathbb F}_{16}$)
The computation of $H^*$ for long codes can be carried often to the computation of $H^*$ for much shorter codes. Let $C(D,mQ)$ be a code and let $n''$ be the largest integer for which equality in the Goppa bound holds. Then $n''<n$ and there exists a divisor $D''\le D$ such that $D''\sim n''Q$. Hence, for $m\ge n$ we have $\ell(mQ-D)=\ell((m-n'')Q-D')$ where $D'=D-D''$. This leads us to considering the codes $C(D',m'Q)$ of length $n'=n-n''$. 
To give an example of this situation let us consider the curve ${\mathcal X}$ over ${\mathbb F}_{16}$ defined by the affine equation
$$
y^{15}=p(x):=\frac{x(x^{14}-1)}{x-1}=x^{14}+x^{13}+\cdots + x.
$$  
Let us study the rational points of $\mathcal X$. Firstly there is just one point $Q$ over $x=\infty$. Regarding the affine points, note that the polynomial $p(x)$ has 2 roots in ${\mathbb F}_{16}$, namely 0 and 1. In fact, if $\alpha\neq 0,1$ is a root of $p(x)$, then $\alpha^7=1$ and $7\nmid 15$. These roots give two points, $R_1=(0,0)$ and $R_1=(1,0)$. We consider now the morphism $\phi=x$, $\phi:{\mathcal X}\rightarrow {\mathbb P}^1(\overline{\mathbb F}_{16})$ of order 15, where $\overline{\mathbb F}_{16}$ denotes the algebraic closure of $\mathbb{F}_{16}$. For $\alpha\in {\mathbb F}_{16}$, $\alpha\neq 0,1$, from the equation of $\mathcal X$, we have $y^{15}=\alpha(\alpha^{14}-1)/(\alpha-1)=1$, so that there are 15 rational points over each $\phi(\alpha)$. Write
$$
\mbox{div}(x-\alpha)=\sum_{i=1}^{15}P^i_\alpha-15 Q. 
$$
Thus $\mathcal X$ has $(16-2)\cdot 15+2+1=213$ rational points. To compute its genus observe that 
$$
y^{15}=x(x-1)(x-\alpha_1)^2\ldots (x-\alpha_6)^2 ,
$$  
where $\alpha_i^7\neq 1$, $\alpha_i\not\in {\mathbb F}_{16}$. 
As the extension ${\mathbb F}_{16}({\mathcal X})|{\mathbb F}_{16}(x)$ is Kummer, the genus can be computed via the Riemann-Hurwitz formula \cite{Sti},
   $$
2g-2=15(-2)+9(14)=96
   $$
and $g=49$. Note that $\mathcal X$ attains the record of rational points among all curves genus 49 over ${\mathbb F}_{16}$. Finally let us compute the  Weierstrass semigroup $H$ at $Q$. We have seen that $-v_Q(x)=15$. In the same way $\mbox{div}_{\infty}(y)=14Q$, so $14,15\in H$. Let
   $$
z:=y^8/((x-\alpha_1)\cdots (x-\alpha_6)).
   $$
It is easy to compute $\mbox{div}_{\infty}(z)=22Q$, hence $22\in H$ and thus $\langle 14,15,22\rangle\subseteq H$. Since both semigroups have equal genus we conclude that equality holds. Then 
   \begin{align*}
H(Q)= & \langle 14, 15, 22\rangle=\{ 0,14,15,22,28,29,30,36,37,42,43,44,45,50,51,52,56,57,58,\\
   {} & 59,60,64,65,66,67, 70,71,72,73,74,75, 
78,79,80,81,82,84,85,86,87,88,89,\\
  {} & 90,92,93,94,95,96,97,98,99,100,101,102,103,104,105,\ldots  \}.
   \end{align*}
Note that $2g-1= 97\in H$ and so $H$ is not symmetric. In order to construct codes from this curve let us consider the divisors $D'=R_1+R_2$, and for $\alpha\in {\mathbb F}_{16}$, $\alpha\neq 0,1$
$$
D''_\alpha=\sum_{i=1}^{15}P^i_\alpha \; , \; D''=\sum_{\alpha} D''_\alpha.
$$
According to our previous computations, $D''_\alpha\sim 15 Q$ and hence $D''\sim 210 Q$. 
Let $D=D'+D''$ be the sum of all affine points of $\mathcal X$,
$n=212=\deg (D)$ and consider the codes of length $n$, $C(D,mQ)$, $m=0,\dots,n+2g-1=309$. In order to determine $H^*=H^*(D,Q)$ we have to compute $\ell(mQ-D)$ for $m\ge n$. But since $D''\sim 210 Q$, then $\ell(mQ-D)=\ell((m-210)Q)-D')$. This fact leads us to considering the codes $C(D',m'Q)$ for $m'=2,\dots,2g+1=99$. The length of these codes is $n'=2$ and $C(D',0Q)=\langle  (1,1) \rangle$. Thus there exists just one $m'$ for which the dimension increases. Clearly, this is not the case for any gap of $H$, so $m'$ must be a non-gap. Looking at the generator matrix $(1,1)$ of $C(D',0Q)$ we conclude that this $m'$ is the smallest order of a function $f$ in $ {\mathcal L} (\infty Q)$ such that $f(R_1)\neq f(R_2)$. Such a function is clearly $f=x$ and hence $m'=15$. Thus,  
$$
H(D,Q)^* \cap \{ n,\dots,n+2g-1\}=\{ n-2+l : \mbox{ $l$ is a gap of $H$ and $l\ge 2$} \} \cup \{ n-2+15 \}.
$$
Once $H^*$ is known we can compute the dimensions of all codes $C(D, mQ)$ and apply Theorem \ref{asterisk} to estimate the minimum distances. Note that for large $m$ we do not obtain good parameters. In fact, as $D''_\alpha\sim 15 Q$, for all $m<n$, $m$ multiple of $15$, the true minimum distance of $C(D,mQ)$ equals the Goppa estimate. In particular the minimum distance distance of $C(D, 210 Q)$ is $d=2$. The bound $d^*$ gives $d\ge 2$ for $m=224$ (that is, for dimension $k\le 175$) and hence all codes $C(D,mQ)$, $m=210,\dots, 224$ have true minimum distance $d=2$.  

In order to obtain codes with better parameters (that is, better minimum distance) the usual approach is to consider another divisor $G$. We shall show that this goal can also be accomplished by taking a slightly different $D$. Consider the codes $C(D'',mQ)$ of length $n''=210$. Then the function from which the codeword of weight 2 arises belongs to the kernel of the evaluation map. The set $H^*=H^*(D'',Q)$ can be now computed by using Corollary \ref{sim}, and 
$H^*\cap\{ n'',\dots,n''+2g-1\}=\{ n''+l_1,\dots,n''+l_g \}$, where $l_1,\dots,l_g$ are the 49 gaps of $H$. It is not necessary to apply the bound $d^*$ to see that the minimum distance of these codes is larger for $m\ge n''$. For example, from the improved Goppa bound we know that the minimum distance of $C(D'', 210 Q)$ satisfies $d\ge n''-210+\gamma_2=\gamma_2$, where $\gamma_2$ is the usual gonality of $\mathcal X$, see \cite{M} . It is not easy to compute $\gamma_2$, but at the first sight we have $\gamma_2\ge \#{\mathcal X}({\mathbb F}_{16})/\# {\mathbb P}^1({\mathbb F}_{16})$, hence $\gamma_2\ge 13$ (so $\gamma_2=13$ or 14) and $d\ge 13$ as well. 
  \end{example}

   \section{Relating the bounds $d^*$ and $d_{ORD}$}\label{se:reb}

As we noted above, in some cases the generalized order bounds may give different results than the original order bound, see \cite{Beelen} Example 8.  Likewise, also the  Andersen-Geil bound,  from which we have obtained $d^*$, can be very different from the original order bound, see Example 51 of \cite{AG}. In this Section we shall compare $d^*$ and the original order bound $d_{ORD}$.  This comparison can be done over sequences of one-point codes such that their duals are also one-point. We can slightly relax this condition by imposing that the duals are isometric to one-point codes.

    \subsection{The isometry-dual condition}

Let $C,D$, be two linear codes in ${\mathbb F}_q^n$ and let ${\mathbf x}\in ({\mathbb F}_q^*)^n$ be an $n$-tuple of non-zero elements. We say that $C$ and $D$ are {\em isometric} according to ${\mathbf x}$ (or simply {\em ${\mathbf x}$-isometric}) if the map $\chi_{\mathbf x}:{\mathbb F}_q^n \rightarrow {\mathbb F}_q^n$ given by $\chi_{\mathbf x} ({\mathbf v})={\mathbf x}*{\mathbf v}$ satisfies $\chi_{\mathbf x}(C)=D$.
Note that $\chi_{\mathbf x}$ is a true  linear isometry for the Hamming distance, hence isometric codes have the same parameters. The dual of a code $C$ is denoted by $C^{\perp}$.

  \begin{proposition}
Let $C,D$ be two linear codes in ${\mathbb F}_q^n$. If $\chi_{\mathbf x}(C)=D$ then $\chi_{\mathbf x}(D^{\perp})=C^{\perp}$.   
  \end{proposition}
  \begin{proof}
Let  ${\mathbf c}\in C$ and ${\mathbf d}=\chi_{\mathbf x}({\mathbf c})\in D$. For all ${\mathbf v}\in {\mathbb F}_q^n$ we have ${\mathbf v}\cdot{\mathbf d}={\mathbf v}\cdot({\mathbf x}*{\mathbf c})= ({\mathbf x}*{\mathbf v})\cdot{\mathbf c}$, hence ${\mathbf v}\in D^{\perp}$ if and only if $\chi_{\mathbf x}({\mathbf v})\in C^{\perp}$. 
  \end{proof}

Let us recall that we have fixed a basis ${\mathcal B} =\{ {\mathbf b}_1, \dots, {\mathbf b}_n\}$ of ${\mathbb F}_q^n$ and the associated codes  $C_i=\langle {\mathbf b}_1,\dots,{\mathbf b}_i \rangle$,  $i=0,\dots,n$.

  \begin{definition}
A sequence of codes $(C_i)_{i=0,\ldots,n}$ is said to satisfy the {\em isometry-dual} condition if there exists ${\mathbf x}\in ({\mathbb F}_q^*)^n$ such that $C_i$ is  ${\mathbf x}$-isometric to $C_{n-i}^{\perp}$ for all $i=0,1,\dots,n$.
  \end{definition}

Let us study  the case of AG codes. We consider the sequence of codes $(C(D,m_iQ))_{i=0,\ldots,n}$ arising from the curve $\mathcal X$  and the associated set $H^*=\{m_1,\dots,m_n\}$. In addition let $m_0=-1$ and $C(D,m_0Q)=(0)$. 
If $(C(D,m_iQ))$ satisfies the isometry-dual condition then both $d^*$ and the order bound $d_{ORD}$ can be used to estimate the minimum distance of these codes.
Let us remember that we are assuming that $n>2g+2$. Remember also that the dual of $C(D,mQ)$ is $C(D, D+W-mQ)$, where $W$ is a canonical divisor with simple poles and residue 1 at every point in supp$(D)$ (see \cite{Sti}). 

  \begin{proposition}
The following statements are equivalent. 
\begin{enumerate}
\item[\rm (a)] The sequence $(C(D,m_iQ))_{i=0,\dots,n}$ satisfies the isometry-dual condition.
\item[\rm (b)] The divisor $(n+2g-2)Q-D$ is canonical.
\item[\rm (c)] $n+2g-1\in H^*$.
\end{enumerate}
  \end{proposition}
  \begin{proof}
 Let us consider the divisor $E=(n+2g-2)Q-D$ and for an integer $m$ write $m^{\perp}=n+2g-2-m$. 
((a)$\Leftrightarrow$(b)) Assume that the sequence $(C(D,m_iQ))_{i=0,\dots,n}$ satisfies the isometry-dual condition. Let  $m$ be such  that $2g\le m \le n-2$ (since $n>2g+2$, such an $m$ does exist). Then $2g\le m^{\perp}\le n-2$ and hence $m,m^{\perp}\in H^*$. In particular $\dim(C(D,mQ))+\dim(C(D,m^{\perp}Q))=n$.
Since the sequence $(C(D,m_iQ))_{i=0,\dots,n}$ satisfies the isometry-dual condition we have that $C(D,D+W-mQ)=C(D,mQ)^{\perp}$ is isometric to $C(D,m^{\perp}Q)$. This implies that the divisors $D+W-mQ$ and $m^{\perp}Q$ are equivalent (see \cite{MP}). Then $W\sim (m+m^{\perp})Q-D=E$ and this divisor is canonical. Conversely, if $E$ is a canonical divisor then there is a rational function $f$ such that $E+\mbox{div}(f)=W$. In particular $f$ has neither poles nor zeros in supp$(D)$. Let ${\mathbf x}=ev_D(f)$. Then we have $D+W-mQ=m^{\perp}Q+\mbox{div}(f)$ hence $C(D,mQ)^{\perp}={\mathbf x}*C(D,m^{\perp}Q)=\chi_{\mathbf x} (C(D,m^{\perp}Q)$. 
((b)$\Leftrightarrow$(c))  Since $\deg(E)=2g-2$, then $E$ is canonical if and only if $\ell (E)=g$. By the Riemann-Roch theorem (see \cite{HLP}, Theorem 2.55), we have $\ell (E+Q)=g$ hence $E$ is canonical if and only if $\ell(E)=\ell(E+Q)$, that is, if and only if $n+2g-1\in H^*$ according to Proposition \ref{ell}
  \end{proof}
  
  \begin{example} (Codes on Castle curves)
Let $\mathcal X$ be a Castle curve and $(C(D,m_iQ))_{i=0,\ldots,n}$ be a sequence of Castle codes of length $n$ arising from $\mathcal X$ (see Example \ref{castle}).  Since $D\sim nQ$ and the semigroup $H(Q)$ is symmetric, Proposition \ref{sim} implies that $(C(D,m_iQ))_{i=0,\ldots,n}$ satisfies the isometry-dual condition. 
  \end{example}  
  
  \begin{example} (Codes on the Klein quartic)
Let us consider the Klein quartic ${\mathcal X}$ of projective equation $X^3Y+Y^3Z+Z^3X=0$ and genus $g=3$. Over the field ${\mathbb F}_8$, ${\mathcal X}$ has 24 rational points (the maximum allowed by Weil-Serre bound) and a rich geometrical structure. Codes coming from this curve are usually constructed by using the divisors $G=m(Q_1+Q_2+Q_3)$, where $Q_1=(1:0:0),Q_2=(0:1:0)$ and $Q_3=(0:0:1)$, since this choice  has some technical advantages (see \cite{Du},\cite{Hansen},\cite{HLP}). However, one-point codes over ${\mathcal X}$ can also be considered.
Let $Q=Q_2$, $D'=Q_1+Q_3$, $D''=P_1+\dots+P_{21}$ be the sum of all rational points except $Q_1,Q_2,Q_3$ and let $D=D'+D''$. 
It is easy to see that div$(x)=3Q_3-2Q_2-Q_1$ and div$(y)=2Q_1+Q_3-3Q_2$. Then div$(xy)=Q_1+4Q_3-5Q_2$ and div$(x^2y)=7Q_3-7Q_2$.
Then the Weierstrass semigroup $H=H(Q)$ is generated by 3,5 and 7. In particular $\{ 1,y,xy,y^2,x^2y,\dots \}$
is a basis of ${\mathcal L}(\infty Q)$.  In order to compute $H^*=H^*(D,Q)$ we can proceed as in Example \ref{ej}.  By considering the morphism $\phi=y$, $\phi:{\mathcal X}\rightarrow {\mathbb P}^1(\overline{\mathbb F}_{8})$ of degree 3, we observe that $D''\sim 21 Q$. This fact leads us to consider the codes $C(D',mQ)$ of length 2 and the set $H^*(D',Q)$. 
Since $x^2y$ is the first non constant function in the above basis for which $Q_1$ is not a zero, we deduce that $H^*(D',Q)=\{ 0,7 \}$. Then $21+7=28=n+2g-1\in H^*(D,Q)$ and the sequence of codes $C(D,m_iQ)$ satisfies the isometry-dual condition. As we shall se in Lemma \ref{symmetry}, this condition provides the whole set $H^*$ and $H^*=\{0,3,5,6,7,\dots,22,23,25,28 \}$. A direct computation shows that for this sequence of codes, both $d^*$ and the order bound give the true minimum distance for all $m$.
  \end{example} 
  
  \begin{example}
Let us consider the sequence of codes of length $n=212$ introduced in Example \ref{ej}. Here $n+2g-1=309\not\in H^*$ hence this sequence does not satisfy the isometry-dual condition. As a consequence $d_{ORD}$ cannot be applied to estimate the minimum distances.  
  \end{example}

    \subsection{The bounds for isometry dual codes}

Let $(C(D,m_iQ))_{i=0,\dots,n}$ be a sequence of one-point codes satisfying the isometry-dual condition.
For this sequence the set $H^*$ is particularly simple and can be computed just in terms of the Weierstrass semigroup $H$. 

  \begin{lemma}\label{symmetry}
If $(C(D,m_iQ))$ satisfies the isometry-dual condition, then $H^*=\{ m\in H : n+2g-1-m\in H  \}$. 
  \end{lemma}
  \begin{proof}
Let $m\in H$. From the Riemann-Roch theorem, $\ell (mQ-D)=m-n+1-g+\ell((n+2g-2-m)Q)$ and hence $\ell(mQ-D)=\ell((m-1)Q-D)$ if and only if $\ell((n+2g-2-m)Q)\neq \ell((n+2g-1-m)Q)$, that is, if and only if $n+2g-1-m\in H$.   
  \end{proof}

Thus for isometry-dual sequences the set $H^*$ is symmetric in the sense that for an integer $m$ it holds that $m\in H^*$ if and only if $n+2g-1-m\in H^*$  (and conversely this property implies the isometry-dual condition). It follows that $n+2g-1-m_i=m_{n-i+1}$. We must not confuse this kind of symmetry with the symmetry of the semigroup $H$.
Let us remember that a semigroup $H$ of genus $g$ is called {\em symmetric} if $2g-1\not\in H$ or equivalently (since its largest gap $l_g$  satisfies $l_g\le 2g-1$) if $l_g=2g-1$. For symmetric semigroups it holds that $m\in H$ if and only if $l_g-m\not\in H$, see \cite{HLP}. When the Weierstrass semigroup $H=H(Q)$ is symmetric, $(2g-2)Q$ is a canonical divisor, hence the isometry-dual property is equivalent to $D\sim nQ$. Since in this case the condition $n+2g-1-m\in H$ is equivalent to $m-n\not\in H$, or $m\not\in n+H$, then the set $H^*$ is given by
$$
H^*=H\setminus (n+H).
$$

Let us return to the general case of $H$, where it might not be symmetric. The symmetrical description of $H^*$ given by Lemma \ref{symmetry} allows us to write $H^*$  in the following way 

  \begin{proposition}
If the sequence $(C(D,m_iQ))$ satisfies the isometry-dual condition, then $H^*=\{ 0,\dots,n+2g-1  \}\setminus \{ l_1,\dots,l_g,n+2g-1-l_g,\dots,n+2g-1-l_1 \}$. 
  \end{proposition}
  \begin{proof} 
We have $l_1,\dots,l_g\not\in H^*$. In the same way, if $l$ is a gap of $H$ then
$n+2g-1-(n+2g-1-l)=l\not\in H$  and hence $n+2g-1-l\not \in H^*$. Furthermore, since $l_g<n$, then $l_g<n+2g-1-l_g$ and hence $\# \{ l_1,\dots,l_g,n+2g-1-l_g,\dots,n+2g-1-l_1 \}=2g$. By cardinality reasons we get the result.
  \end{proof}

For $i=1,\dots,n$, let us consider the set $L_i= \{ m_i+l_1,\dots, m_i+l_g \}$.

  \begin{proposition}\label{LambdaL}
If $(C(D,m_iQ))$ satisfies the isometry-dual condition, then $\# \Lambda^*_i=n-i+1-\# (L_i\cap H^*)$.
  \end{proposition}
  \begin{proof}
Let $L=\{ l_1,\dots,l_g,n+2g-1-l_g,\dots,n+2g-1-l_1 \}$, and for $i=1,\dots,n$,
\begin{eqnarray*}
B_i^{(1)} &=& \{ m_j\in H^* : m_i+m_j< n+2g, m_i+m_j\not\in H^* \},   \\
B_i^{(2)} &=& \{ m_j\in H^* : m_i+m_j\ge n+2g \} . 
\end{eqnarray*} 
Clearly $\# \Lambda^*_i=\#(H^*\setminus (B_i^{(1)}\cup B_i^{(2)}))=n-\# B_i^{(1)}-\# B_i^{(2)}$. Since $H^*\subseteq H$ and the sum of two non-gaps is again a non-gap, we have $B_i^{(1)}=\{ m_j\in H^* : m_i+m_j \in L \}= \{ n+2g-1-l_g-m_i,\dots,n+2g-1-l_1-m_i \}\cap H^*$. According to Lemma \ref{symmetry},  $\# B_i^{(1)}= \# (L_i\cap H^*)$. Besides $\# B_i^{(2)}=i-1$. In fact, if  $m_i+m_j\ge n+2g$, from Lemma \ref{symmetry} we can write $m_j=n+2g-1-m_t$ with $t=n-j+1$. Then $n+2g-1+m_i-m_t>n+2g-1$ if and only if $m_i>m_t$ and there exist $i-1$ such choices for $m_t$.
  \end{proof}

Then $d^*$ can be written for isometry-dual codes  as
$$
d(C(D,m_iQ))\ge d^*(i)=\min \{ n-r+1-\# (L_r\cap H^*) : r\le i  \}.
$$
 
Let us prove now that $d^*$ and the strict order bound with respect to the evaluation map $ev_D$, $d_{ORD,ev}$ (\cite{HLP}, Section 4.3), give the same result when applied to codes satisfying the isometry-dual condition. Let $m_i\in H^*$ and let us compute both bounds for $C(D,m_iQ)$. If $m_i< n-l_g$, according to Proposition \ref{goppa} and Theorem 4.7 in \cite{HLP}, both bounds are equal to Goppa bound.

In order to compute the order bound, we first need the duals of the codes $C(D,m_rQ)$. As we know, $C(D,m_rQ)^{\perp}$ is isometric to $C(D,(n+2g-2-m_r)Q)$. Let $h_s,h_{s+1}\in H$ be such that $h_s\le n+2g-2-m_r<h_{s+1}$. Then $C(D,h_sQ)=C(D,(n+2g-2-m_r)Q)$ and hence $C(D,h_sQ)^{\perp}$ is isometric to $C(D,m_rQ)$. Note that $C(D,m_rQ)$ has dimension $r$, so  $C(D,h_sQ)$ has dimension $n-r$. Furthermore, Lemma \ref{symmetry} implies that $n+2g-1-m_r\in H^*$ hence $h_{s+1}=n+2g-1-m_r=m_{n-r+1}$ and $\dim C(D,h_{s+1}Q)=n-r+1$.

For $h\in H$ let us consider the set 
$$
A[h]=\{ t\in H : h-t\in H \}.
$$
The strict order bound on the minimum distance of $C(D,m_iQ)$   together with our previous discussion, imply that 
\begin{eqnarray*}
d(C(D,m_iQ))  \ge  d_{ORD,ev}(C(D,m_iQ)) &:=& \min \{ \# A[h] : h\in H^*, h\ge n+2g-1-m_i \} \\
          & =   & \min\{ \# A[n+2g-1-m_r] : m_r\in H^*, r\le i   \} \\
          & =   & \min\{ \# A[m_{n-r+1}] : r\le i   \},
\end{eqnarray*}
where the last two equalities follow from \ref{symmetry} and the fact that $m_{n-r+1}=n+2g-1-m_r$.

  \begin{lemma}\label{lagunas}
If $h\in H$ and $l\not\in H$ then $l-h\not\in H$. 
  \end{lemma}
  \begin{proof}
If $l-h=h'\in H$ then $l=h+h'$ and hence $l\in H$. 
  \end{proof}

  \begin{proposition}\label{LambdaA}
Let $m_r\in H^*$. If $(C(D,m_iQ))$ satisfies the isometry-dual condition, then $\# \Lambda^*_r=\# A[m_{n-r+1}]$.
  \end{proposition}
  \begin{proof}
Let us compute $\# A[n+2g-1-m_r]+\# (L_r\cap H^*)$. For a given gap $l$ of $H$, we  have $m_r+l\in H^*$ if and only if $n+2g-1-m_r-l\in H^*$. Thus 
\begin{eqnarray*}
\# (L_r\cap H^*)&=& \# \{ l\in \mbox{Gaps}(H) : n+2g-1-m_r-l\in H^*  \}\\
                &=& \# \{ h\in H^* : n+2g-1-m_r-h\in \mbox{Gaps}(H)  \},
\end{eqnarray*}
so $\# A[n+2g-1-m_r]+\# (L_r\cap H^*)= \# \{ h\in H : h\le n+2g-1-m_r \}-\# \{ h\in H\setminus H^* : h\le n+2g-1-m_r,  n+2g-1-m_r-h \in \mbox{Gaps}(H) \}$. Let us note that for all $h\in H\setminus H^*$, $h\le n+2g-1-m_r$, it holds that  $n+2g-1-h\in \mbox{Gaps}(H)$. In fact, according to Lemma \ref{symmetry}, we would otherwise have $h\in H^*$. Then, from Lemma \ref{lagunas}, $n+2g-1-m_r-h \in \mbox{Gaps}(H)$. So $\{ h\in H\setminus H^* : h\le n+2g-1-m_r, n+2g-1-m_r-h \in \mbox{Gaps}(H) \}=\{ h\in H\setminus H^* : h\le n+2g-1-m_r \}$ and hence $\# A[n+2g-1-m_r]+\# (L_r\cap H^*)= \# \{ h\in H^* : h\le n+2g-1-m_r \}=\dim (C(D,(n+2g-1-m_r)Q))=n-r+1$. 
  \end{proof}

  \begin{corollary}
For isometry-dual codes, we have $d_{ORD,ev}(C(D,m_iQ))=d^*(i)$.
  \end{corollary}

Therefore $d^*$ and the strict order bound are the same for isometry-dual codes.

\subsection{More on improved codes}

In Section \ref{Sect_improved} we have considered the improved code $C(D,Q,\delta)=\langle \{ ev(f_i) : \# \Lambda^*_i\ge \delta \}\rangle$, for $1\le \delta\le n$. It is analogous to the improved code $\tilde{C}(D,Q,\delta)$ introduced by Feng and Rao, \cite{FR,HLP}, based on the order bound:  
$$
\tilde{C}(D,Q,\delta)=\langle \{ ev (f_i) : \# A[m_i]<\delta \}\rangle^{\perp}.
$$  
It is well known that the minimum distance of $\tilde{C}(D,Q,\delta)$ is at least $\delta$. When the sequence $(C(D,m_i Q))$ is isometry-dual, Proposition \ref{LambdaA} allows us to write $\tilde{C}(D,Q,\delta)$ in terms of the sets $\Lambda^*_i$'s,
$$
\tilde{C}(D,Q,\delta)=\langle \{ ev (f_i) : \# \Lambda^*_{n+1-i} \ge \delta \}\rangle^{\perp}.
$$
Then it is natural to wonder about the relation between these two improved codes $\tilde{C}(D,Q,\delta)$ and $C(D,Q,\delta)$. 

  \begin{proposition}
If the sequence $(C(D,m_iQ))$ satisfies the isometry-dual condition, then  $C(D,Q,\delta)$ and $\tilde{C}(D,Q,\delta)$ have the same dimension. 
  \end{proposition}
  \begin{proof}
If $C(D,Q,\delta)$ is generated by $t$ vectors then $\tilde{C}(D,Q,\delta)$ is defined by $n-t$ independent parity checks.  
  \end{proof}
  
If the sequence $(\Lambda^*_i)$ is monotone for $\delta$ then $C(D,Q,\delta)$ is a one-point code, hence $C(D,Q,\delta)$ and $\tilde{C}(D,Q,\delta)$ are isometric. Let us study the general case.

  \begin{lemma}\label{lem1}
Let $(C_i=\langle {\mathbf b}_1, \dots , {\mathbf b}_i\rangle)$ be a sequence  of codes that satisfies the isometry-dual condition, $\chi_{\mathbf x}(C_i)=C_{n-i}^{\perp}$. Then for $i=1,2,\dots,n$, we have
$$
\chi_{\mathbf x}({\mathbf b}_i)\in C_{n-i}^{\perp}\setminus C_{n-i+1}^{\perp}.
$$
  \end{lemma}
  \begin{proof}
Follows directly from the definition of isometry-dual sequence.  
  \end{proof}
 
Let us remember that the improved codes $C(D,Q,\delta)$ and $\tilde{C}(D,Q,\delta)$ depend on the choice of functions $f_1,\dots,f_n$ in ${\mathcal{L}}(\infty Q)$ such that $v(f_i)=m_i$.

  \begin{lemma}\label{lem2}
If $( C(D,m_iQ))$ satisfies the isometry-dual condition then given a set $\{f_1, \ldots , f_n\}$ of functions in
${\mathcal{L}}(\infty Q)$ with $v(f_i)=m_i$,  there exists a similar set $\{f_1^\prime, \ldots , f_n^\prime\}$ such that $\chi_{\mathbf x}(ev(f_i^\prime)) \cdot ev(f_j) \neq 0$ holds if and only if $j=n-i+1$.
  \end{lemma}
  \begin{proof}
By Lemma~\ref{lem1} and the isometry-dual condition, the sets $\{f_1, \ldots ,f_n\}$ and  $\{f_1^\prime, \ldots , f_n^\prime\}$ will satisfy
$$
\begin{array}{ll}
(\chi_{\mathbf x}(ev(f_i^\prime)))\cdot ev(f_j)=0& {\mbox{for \ }}j=1,\ldots ,n-i\\
(\chi_{\mathbf x}(ev(f_i^\prime)))\cdot ev(f_{n-i+1})\neq 0.
\end{array}
$$
So, we have to determine a particular set $\{f_1^\prime, \ldots,f_n^\prime\}$ that in addition satisfies
\begin{eqnarray}
(\chi_{\mathbf x}(ev(f_i^\prime)))\cdot ev(f_j)=0& {\mbox{ \ for \ }} &
j=n-i+2, \ldots , n.\label{eqbox}
\end{eqnarray}
We show the existence of such a set by induction. Note first that given arbitrary $\{f_1, \ldots , f_n\}$ then the condition~(\ref{eqbox}) is trivially satisfied for $i=1$ if we choose $f_1^\prime=f_1$. Assume next that (\ref{eqbox}) holds for all values of $i=1, \ldots , s$, where $s$ is some number less than $n$. That is, for each $i \in \{
1, \ldots , s\}$ the only $j$ such that $\chi_{\mathbf x}(ev(f_i^{\prime}))\cdot ev(f_j)\neq 0$ is $j=n-i+1$. Denote by $a_j$ the value of $\chi_{\mathbf x}(ev(f_i^{\prime}))\cdot ev(f_j)$, $j=n-s+1, \ldots ,n$. The function
$$
f_{s+1}^{\prime}=f_{s+1}-\sum_{i=1}^s
  \frac{a_i}{\chi_{\mathbf x}(ev(f_i^{\prime}))\cdot ev(f_{n+1-i})}f_i^{\prime}
$$
satisfies (\ref{eqbox}) as 
\begin{eqnarray*}
\lefteqn{\chi_{\mathbf x}(ev(f_{s+1}^{\prime}))\cdot ev(f_j)=}\\
&& \chi_{\mathbf x}(ev(f_{s+1})) \cdot ev(f_j)-\bigg(\sum_{i=1}^s
   \frac{a_i}{\chi_{\mathbf x}(ev(f_i^{\prime}))\cdot
   ev(f_{n-i+1})}\chi_{\mathbf x}(ev(f_i^{\prime})) \cdot ev(f_j) \bigg).
\end{eqnarray*}
\end{proof}

  \begin{proposition}
Assume $\big( C(D, m_1 Q)\big)$ satisfies the isometry-dual condition. For every choice of $\{f_1, \ldots , f_n\}$ of functions in ${\mathcal{L}}(\infty Q)$ with $v(f_i)=m_i$, there exists a similar set $\{f_1^\prime, \ldots  f_n^\prime\}$ such that the code $\tilde{C}(D,Q,\delta)$ defined from the first set is
isometric to the code $C(D,Q,\delta)$ defined from the latter set. A similar result holds the other way around.
  \end{proposition} 

  \begin{proof} 
By Propositions \ref{LambdaL} and \ref{LambdaA} we have  $\# \Lambda_i^\ast=\#A[m-{n+1-i}]$. Choosing $\{f_1^\prime, \ldots ,f_n^\prime\}$ such that Lemma~\ref{lem2} is satisfied and applying the definitions of $C(D,Q,\delta)$ and $\tilde{C}(D,Q,\delta)$ proves the first claim. The last claim follows by symmetry. 
  \end{proof}

\section{Generalized Hamming weights} \label{se:ghw}

The same ideas used to obtain the bound $d^*$ for the minimum distance can be applied to all generalized Hamming weights (see \cite{AG}). Let us remember that given a set $D \subseteq {\mathbb F}_q^n$, the {\em support} of $D$ is defined as 
$$
\mbox{supp}(D)=\bigcup_{{\mathbf v}\in D} \mbox{supp}({\mathbf v}).
$$
Let $C$ be a code of dimension $k$. For $r=1,\dots,k$, the $r$-th {\em generalized Hamming weight} of $C$ is defined as
$$
d_r(C)=\min \{\# \mbox{supp}(D) : D \mbox{ is an $r$-dimensional linear subspace of $C$} \},
$$
and the sequence $d_1(C),\dots,d_k(C)$, is called the {\em weight hierarchy} of $C$.
Let us first look a general bound on the $d_r(C)$'s. Recall that we have a basis  ${\mathcal B} =\{ {\mathbf b}_1, \dots, {\mathbf b}_n\}$ of ${\mathbb F}_q^n$ and codes $C_i=\langle {\mathbf b}_1,\dots,{\mathbf b}_i \rangle$. 

  \begin{lemma}\label{dr}
Let $D \subseteq {\mathbb F}_q^n$ be a linear subspace of dimension $r$ and let $\{ {\mathbf c}_1,\dots,{\mathbf c}_r \}$ be a basis of $D$. Then $\#\mbox{supp}(D)\ge\#  \cup_{i=1,\dots,r}\{ \nu({\mathbf b}_{\nu({\mathbf c}_i)}*{\mathbf b}_j) : j\in \Lambda_{\nu({\mathbf c}_i)} \}$.
  \end{lemma}
  \begin{proof}
Given $D$, let us consider the space $V(D)=\{ {\mathbf v}\in {\mathbb F}_q^n : \mbox{supp}({\mathbf v})\subseteq \mbox{supp}(D) \}$. Since $\#\mbox{supp}(D)=\dim (V(D))$ and $\mbox{supp}(D)=\mbox{supp}({\mathbf c}_1)\cup\dots\cup\mbox{supp}({\mathbf c}_r)$, we have that $V(D)=V({\mathbf c}_1)+\dots + V({\mathbf c}_r)$ and the statement follows from the results in Section 2.
  \end{proof}

  \begin{theorem}
For $r=1,\dots,i$, the $r$-th generalized Hamming weight of $C_i$ satisfies
$$
d_r(C_i)\ge \min_{1\le j_1<\dots< j_r\le i} \# \left \{\bigcup_{j\in\{ j_1,\dots,j_r \}}
\{\nu({\mathbf b}_j*{\mathbf b}_t) : t\in\Lambda_j  \}\right \}.
$$
  \end{theorem}
  \begin{proof}
 According to Lemma \ref{properties} (c), every linear subspace $D$ of $C_i$ has a basis $\{ {\mathbf c}_1,\dots,{\mathbf c}_r \}$ such that $1\le \nu({\mathbf c}_1)<\dots<\nu({\mathbf c}_r)\le i$. Conversely, given vectors $\{ {\mathbf c}_1,\dots,{\mathbf c}_r \}$ satisfying the above condition, $\langle {\mathbf c}_1,\dots,{\mathbf c}_r \rangle$ is  a vector subspace of $C_i$ of dimension $r$. Then the result is a consequence of Lemma \ref{dr}. 
  \end{proof}

This result is easily translated to one-point AG codes. With the notation as in Section 3, we have codes $C(D,mQ)$ and $C_i=C(D,m_iQ)$. We showed that  
$\#\{\nu({\mathbf b}_j*{\mathbf b}_t) : t\in\Lambda_j  \})\ge \# \Lambda^*_j$. Thus we have

  \begin{theorem}
Let $m$ be a non-negative integer. For $r=1,\dots,i=\dim (C(D,mQ))$, the $r$-th generalized Hamming weight of $C(D,mQ)$ satisfies
$$
d_r(C(D,mQ))\ge d^*_r(i):=\min_{1\le j_1<\dots<j_r\le i} \# (\Lambda^*_{j_1}\cup\dots\cup\Lambda^*_{j_r}). 
$$ 
  \end{theorem} 
  
This result is similar to the corresponding one for the order bound in \cite{HP}. Also similar results to the ones contained in this section can be obtained for improved codes as well.

{\bf Acknowledgments}. The authors wish to thank Peter Beelen and Tom H{\o}holdt for hospitality and interesting discussions on the subject. This paper was written in part during a visit of the second author to Aalborg University and The Technical University of Denmark. He wishes to thank both institutions for hospitality and support. We also wish to thank Iwan M. Duursma, Radoslav Kirov and Seungkook Park for supporting us with the idea behind the material in Section \ref{sec:BDKP}.

\end{document}